\documentclass{article}
\usepackage[english]{babel}
\usepackage{amsmath,amsfonts,amsthm,amssymb,amscd,natbib}

 \usepackage{verbatim}

\newcommand{\p}{\partial}
\newcommand{\e}{\varepsilon}

\newcommand{\R}{{\mathbb R}}

\newcommand{\T}{{\mathbb T}}
\newcommand{\N}{{\mathbb N}}

\newcommand{\FF}{{\cal F}}

\newcommand{\KK}{{\cal K}}

\newcommand{\PP}{{\cal P}}
\newcommand{\RR}{{\cal R}}

\newcommand{\ph}{\varphi}
\newcommand{\Z}{ \mathbb{Z}}
\newcommand{\dd}{{\textup d}}

\newcommand{\diver}{\mathop{\rm div}\nolimits}
\newcommand{\rot}{\mathop{\rm rot}\nolimits}

\theoremstyle{plain}
\newtheorem{theorem}{Theorem}[section]
\newtheorem{lemma}[theorem]{Lemma}%[section]
\newtheorem{proposition}[theorem]{Proposition}

\newtheorem{definition}[theorem]{Definition}%[section]

\theoremstyle{remark}

\newtheorem{example}[theorem]{Example}

\numberwithin{equation}{section}
\begin{document}
\author{Hayk Nersisyan}
\date{}
\title{Controllability of 3D incompressible Euler equations by a finite-dimensional external
force}
\date{}

\date{}
 \maketitle
\begin{center}

CNRS (UMR 8088), D\'epartement de Math\'ematiques\\
Universit\'e de Cergy--Pontoise, Site de Saint-Martin\\
            2 avenue Adolphe Chauvin\\
       95302 Cergy--Pontoise Cedex, France\\
       E-mail: Hayk.Nersisyan@u-cergy.fr
 \end{center}

\vspace{15 pt}

 {\small\textbf{Abstract.} In this paper, we study the
control system associated with the incompressible 3D Euler system.
We show that the velocity field and pressure of the fluid  are
exactly controllable in projections by the same finite-dimensional
control. Moreover, the velocity is approximately controllable.
 We also prove that 3D Euler
system is not exactly controllable by a finite-dimensional
external force.
  }\\\\

\section{Introduction}

Let us consider the controlled incompressible 3D Euler system:
\begin{align}
\dot{u}+\langle u,\nabla\rangle u+\nabla p&=h+\eta,\,\,\, \diver
u=0,
\label{Pr.E1:1}\\
u(0,x)&=u_0(x).\label{Pr.E1:2}
\end{align}
where $u=(u_1,u_2,u_3)$ and $p $  are unknown velocity field and
pressure of the fluid, $h$ is a given  function, $u_0$ is an
initial condition, $\eta$ is the control taking values in a
finite-dimensional space $E$, and
\begin{equation}
\langle u,\nabla\rangle
v=\sum_{i=1}^3u_i(t,x)\frac{\p}{\p{x_i}}v.\nonumber
\end{equation}
We assume that space variable $x=(x_1,x_2,x_3)$ belongs to the 3D
torus $\T^3=\R^3/2\pi\Z^3$.

The question of global well-posedness of 3D Euler system continues
to be one of the most challenging problems of fluid mechanics.
However, the local existence of solutions is well known (e.g., see
\cite{tey1, tem}). Moreover, Beale, Kato and A. Majda \cite{BKM}
proved  that under the condition
$$\int_0^T \|\rot
u(t)\|_{L^\infty}\dd t<\infty$$ the smooth  solution  exists up to
time $T$.

In this paper, we show that for an appropriate choice of $E$, the
problem is exactly controllable in projections, i.e., for any
finite-dimensional subspaces $F, G \subset H^k $  and for any
$\hat u \in F, \hat p\in G$ there is an $E$-valued control $\eta$
such that problem (\ref{Pr.E1:1}), (\ref{Pr.E1:2}) has a solution
$(u,p)$   on   $[0,T]$ whose  projection onto $F\times G$
coincides with $(\hat u,\hat p)$ at time $T$. We also prove that
the velocity $u$ is approximately controllable, i.e., $u(T)$ is
arbitrarily close to $\hat u $. From Eq.~(\ref{Pr.E1:1}) it
follows that the pressure can be expressed in terms of the
velocity, so we can not expect to control approximately the
pressure and the velocity simultaneously.  The proofs of these
results are based on a development of some ideas from \cite{agr1,
agr2, shi1,shi2}.

Let us mention some earlier results on the controllability of the
Euler and Navier--Stokes systems. The exact controllability of
Euler and Navier--Stokes systems with  control supported by a
given domain  was studied by Coron  \cite {cor},  Fursikov and
Imanuvilov \cite{fuim},  Glass \cite {gla}, and Fern\'{a}ndez-Cara
 et al. \cite{fpgi}. Agrachev and Sarychev \cite{agr1, agr2} were first to study
controllability properties of some PDE's of fluid dynamics  by
finite-dimensional external force. They proved the controllability
of 2D Navier--Stokes and 2D Euler equations. Rodrigues  \cite{rod}
used Agrachev--Sarychev method  to prove controllability of 2D
Navier--Stokes equation on the rectangle with Lions boundary
condition. Later Shirikyan \cite{shi2} generalized this method to
the case of not well-posed equations. In particular, the
controllability of 3D Navier--Stokes equation is proved.

Notice that the above papers concern the problem of
controllability of the velocity.  In this paper, we first develop
the ideas of these works to get the controllability of the
velocity of 3D Euler system. One of the main difficulties  comes
from the fact that the resolving operator of the system is not
Lipschitz continues  in the phase space.   We next deduce the
controllability of the pressure   from   that of the velocity with
the help of an appropriate correction of the control function.

We also treat the question of exact controllability of 3D Euler
equation. In \cite{shi3}, Shirikyan shows that the set of
attainability $A_T(u_0)$ of 2D Euler equation from initial data
$u_0\in C^s$ at time $T>0$ cannot contain a ball of $C^s$. We show
that the ideas of \cite{shi3} can be generalized
 to prove that the set $A(u_0)=\cup_T A_T(u_0)$  also does not contain a
 ball in 3D  case. In particular, 3D Euler
equation is not exactly controllable.

The paper is organized as follows. In Section \ref{S:1}, we give a
perturbative result for 3D Euler system. In Sections \ref{S:2} and
\ref{S:3}, we formulate the main results of this paper, which are
proved in Sections \ref{S:4} and \ref{S:5}. Section \ref{S:6} is
devoted to the problem of exact controllability.
\\

\textbf{Acknowledgments.}   I want to thank Armen Shirikyan for
many fruitful suggestions and discussions.
\\
\newline

\textbf{Notation.} We set
$$H=\{u\in L^2: \diver u=0,\quad \int_{\T^3}u(x)\dd x=0\}.$$
Let us denote   by $\Pi$ the orthogonal projection  onto $H$
in $L^2$. Let $ H^k$ be the space of vector functions $u = (u_1,
u_2, u_3)$ with components in  the Sobolev space of order $k$, and
let $\|\cdot\|_k$ be the corresponding norm. Define $H^k_\sigma: =
H^k\cap H$. The Stokes operator is denoted  by $L := -\Pi \Delta$,
$D(L)=H_\sigma^2$. For any vector $n=(n_1,n_2,n_3)\in \R^3$ we
denote $|n|:=|n_1|+|n_2|+|n_3|$.

Let $J_T:=[0, T ]$ and $X$ be a Banach  space endowed with the
norm $\|\cdot\|_X$. For $1\leq p<\infty$  let $L^p(J_T,X)$ be the
space of measurable functions $u: J_T \rightarrow X$ such that
\begin{equation}
\|u\|_{L^p(J_T,X)}:=\bigg(\int_{0}^T \|u\|_X^p\dd s
\bigg)^{\frac{1}{p}}<\infty.\nonumber
\end{equation}
   The space of continuous functions $u: J_T
\rightarrow X$ is denoted by $C(J_T,X)$.

\section{Perturbative result on solvability of the 3D Euler system}\label{S:1}
Let us consider the Cauchy problem for Euler system on the 3D
torus:
\begin{align}
\dot{u}+\langle u,\nabla\rangle u+\nabla p&=f(t),\,\,\, \diver
u=0,
\label{E1:eul}\\
u(0,x)&=u_0(x).\label{E1:eulic}
\end{align}
 System (\ref{E1:eul}),
(\ref{E1:eulic}) is equivalent to the  problem (see \cite[Chapter
17]{tey1})
\begin{align}
\dot{v}+B (v)&=\Pi f(t),\,\,\,\nonumber\\
v(0,x)&=\Pi u_0(x),\nonumber
\end{align}
where $v=\Pi u$, $B(a,b)=\Pi\{\langle a,\nabla\rangle b\}$ and
$B(a)=B(a,a)$. We shall need the following standard estimates for
the bilinear form $B$:
\begin{align}
\|B(a,b)\|_k\leq C\|a\|_k\|b\|_{k+1}  \qquad &\text{for  }
k\ge 2\label{E1:B1},\\
|(B(a,b),L^k b)|\leq C\|a\|_k\|b\|^2_{k}  \qquad &\text{for  }
k\ge 3,\label{E1:B2}
\end{align}
for any $a\in H_\sigma^k$ and $ b \in H_\sigma^{k+1}$ (see
\cite{CF88}).

Let us consider the problem
\begin{align}
\dot{u}+B(u+\zeta)&=f(t),
\label{E1:pert2}\\
u(0,x)&=u_0(x).\label{E1:ic2}
\end{align}
%where $\zeta\in C^\infty(J_T,H^{k-1})$, $f\in C^\infty(J_T,H^{k-1}) $ and
%$\|u_0\in H^{k-1}\|$ .
\begin{theorem}\label{T:pert}
 Let $T>0$ and $k\ge 4$. Suppose that for some functions  $v_0\in H^{k}_\sigma$,   $\xi\in L^2
(J_T,H^{k+1}_\sigma)$ and $g\in L^1(J_T,H^{k}_\sigma)$ problem
(\ref{E1:pert2}), (\ref{E1:ic2}) with $u_0=v_0$, $\zeta=\xi$ and
$f=g$ has a solution $v\in C(J_T,H^{k}_\sigma)$. Then there are
positive constants $\delta$ and $C$ depending only on the quantity
\begin{equation}
 \|v\|_{C(J_T,H^{k})}+\|\xi\|_{L^2(J_T,H^{k})}
   \nonumber
\end{equation}
such that the following statements hold.

\begin{enumerate}
\item[(i)] If $u_0\in H^{k}_\sigma,$ $\zeta\in L^2(J_T,H^{k+1}_\sigma)$ and $f\in L^1(J_T,H^{k}_\sigma)$  satisfy the
inequalities
\begin{equation}\label{E1:delt2}
\|v_0-u_0\|_{k-1}< \delta, \quad \|\zeta-\xi\|_{L^2(J_T,H^{k})}< \delta, \quad \|f-g\|_{L^1(J_T,H^{k-1})}<
\delta,
\end{equation}
then problem (\ref{E1:pert2}), (\ref{E1:ic2}) has a unique
solution $ u\in C(J_T,H^{k}_\sigma).$
\item[(ii)] Let $$\RR:H^{k}_\sigma\times L^2 (J_T,H^{k+1}_\sigma)\times L^1(J_T,H^{k}_\sigma)\rightarrow C(J_T,H^{k}_\sigma)$$
be the operator that takes each triple $(u_0,\zeta,f )$ satisfying
(\ref{E1:delt2}) to the solution $u$ of (\ref{E1:pert2}),
(\ref{E1:ic2}). Then
\begin{align}
&\|\RR(u_0,\zeta,f )-\RR(v_0,\xi,g )\|_{C(J_T,H^{{k-1}})}\le
C\big( \|v_0-u_0\|_{k-1}\nonumber\\&+
\|\zeta-\xi\|_{L^2(J_T,H^{k})}+\|f-g\|_{L^1(J_T,H^{k-1})}\big).\nonumber
\end{align}
\item[(iii)] Let $\zeta\in C(J_{T},H^k)$ and $f\in C(J_{T},H^{k-1})$, and let $\RR_t$ be the restriction of $\RR$ to the time $t$. Then $\RR_\cdot$ is Lipschitz-continuous in time, i.e.,
\begin{align}
&\|\RR_t(u_0,\zeta,f )-\RR_s(u_0,\zeta,f )\|_{k-1}\le
M|t-s|,\nonumber
\end{align}
where $M$ depends on  $\|\RR(u_0,\zeta,f )\|_{C(J_T,H^{k})},
\|\zeta\|_{C(J_T,H^{k})}$ and $\|f\|_{C(J_T,H^{k-1})}$.
\end{enumerate}
\end{theorem}
\begin{proof}
We seek  a solution of (\ref{E1:pert2}), (\ref{E1:ic2}) in the
form $u=v+w$. Substituting this into (\ref{E1:pert2}),
(\ref{E1:ic2}) and performing some transformations, we obtain the
following problem for $w$:
\begin{align}
&\dot{w}+B(w+\eta,v+\xi)+B(v+\xi,w+\eta)
+B(w+\eta)=q(t,x),\label{E1:pert151}\\
&w(0,x)=w_0(x),\label{E1:eqwin}
\end{align}
where   $w_0=u_0-v_0$, $\eta=\zeta-\xi$ and $q=f-g$.
 By bilinearity
 of $B$,  (\ref{E1:pert151}) is equivalent to the   equation
\begin{align}
&\dot{w}+B(w)+\tilde{B}(w,\eta)+\tilde{B}(w,v)+\tilde{B}(w,\xi)=q(t,x)-(B(\eta)+\tilde{B}(v,\eta)+\tilde{B}(\xi,\eta)),\label{E1:pert5}
\end{align}
where $\tilde B(u,v)=B(u,v)+B(v,u) $. It follows from   (\ref{E1:delt2}) that we can choose  $\delta
>0$ such that the right-hand side of
(\ref{E1:pert5}) and initial data $w_0$ are    small in
$L^1(J_T,H^{k-1}_\sigma)$ and $H^{k-1}_\sigma$, respectively.
Hence, by the standard theorem of existence (see \cite{tey1},
\cite{tem}), system (\ref{E1:pert5}), (\ref{E1:eqwin}) has a
unique  solution $w\in C(J_T,H^{k-1}_\sigma)$. From the embedding
$H^2_\sigma \hookrightarrow L^\infty$ we deduce that
\begin{align}\label{E:BKMLinf}
\sup_{t\in[0,T]}\|\rot u(t,\cdot)\|_{L^\infty}<\infty.
\end{align}
In view of   $u_0\in H^{k}_\sigma$, $\zeta\in
L^2(J_T,H^{k+1}_\sigma)$, $f\in L^1(J_T,H^{k}_\sigma)$ and
(\ref{E:BKMLinf}), the Beale--Kato--Majda theorem
  (see \cite{BKM}) implies $u\in C(J_T, H^{k}_\sigma)$.

To prove $(ii)$, let us get an a priori
estimate for $w$. Multiplying  (\ref{E1:pert5}) by $L^{{k-1}}w$ and using
(\ref{E1:B1}), (\ref{E1:B2}), we obtain
\begin{align}
\frac{1}{2}\frac{d}{dt}\|w\|^2_{{k-1}}&\leq
C\bigg(\|w\|^3_{{k-1}}+ \|w\|_{{k-1}}^2\big(
\|\eta\|_{k}+\|v\|_{k} +\|\xi\|_{k}\big)\nonumber\\
&\quad+\|w\|_{{k-1}}\big( \|q\|_{k-1}+
\|\eta\|_{k}(\|\eta\|_{k-1}+\|v\|_{k}+\|\xi\|_{k}) \big)\bigg)
.\label{E1:14}
\end{align}
%where $H(t)=C(\|q(t,\cdot)\|_{k-1}+\|\eta(t,\cdot)\|_{k})$.
Integrating (\ref{E1:14}), we obtain
\begin{align}
 \|w\|^2_{C(J_t,H^{k-1})}& \le C\|w\| _{C(J_t,H^{k-1})}\nonumber\\&\quad\times\bigg[A+\int _0^t \bigg(\|w\|^2_{{k-1}}+
\|w\|_{{k-1}}\big( \|\eta\|_{k}+\|v\|_{k}+\|\xi\|_{k}
\big)\bigg)\dd t
    \bigg],\label{E1:15}
\end{align}
where
 $A=\|w_0\|_{k-1}+\int_0^T\big [
\|q\|_{k-1}+ \|\eta\|_{k}(\|\eta\|_{k-1}+\|v\|_{k}+\|\xi\|_{k})
\big]
 \dd s$.
Dividing (\ref{E1:15}) by $ \|w\| _{C(J_t,H^{k-1})} $ and using
the Gronwall inequality, we get
\begin{equation}
\|w\|_{{k-1}}\leq A_1+ C_1\int_0^t \|w(s,\cdot)\|_{{k-1}}^2\dd s
,\nonumber
\end{equation}
where $A_1=C_1A$ and $C_1$ is a constant depending on
$\|v\|_{C(J_T,H^{k})}+\|\xi\|_{L^2(J_T,H^{k})}$.
 Another application of Gronwall inequality  gives that
\begin{align}\label{E1:himnw1}
\|w(s)\|_{{k-1}}\le \frac{A_1}{1-C_1A_1t}\le 2A_1 \,\,\,\text{for
any }\,t\le\frac{1}{2C_1A_1}.
\end{align}
We can choose $\delta>0$ such that $\frac{1}{2C_1A_1}\geq T$. From
the definition of $A_1$ and  (\ref{E1:himnw1}) we deduce that
\begin{align}\label{E1:himnw}
\|w\|_{C(J_T,H^{k-1})}\le C
(\|w_0\|_{H^{k-1}}+\|\eta\|_{L^2(J_T,H^{k})}+\|q\|_{L^1(J_T,H^{k-1})}).
\end{align}
 Statement $(ii)$  is a straightforward consequence of (\ref{E1:himnw}).

Let us prove $(iii)$. Integrating (\ref{E1:pert2}) over $(s,t)$
and using (\ref{E1:B1}), we  get
\begin{align}
\|u(t)-u(s)\|_{k-1}\le
\int_s^t\|f(\tau)-B(u(\tau)+\zeta(\tau))\|_{k-1}\dd \tau\le
M|t-s|.\nonumber
\end{align}
This completes the proof of Theorem \ref{T:pert}.
\end{proof}

\section{Controllability of the velocity }\label{S:2}
 Let us consider the controlled Euler system:
\begin{align}
\dot{u}+B (u)&=h(t)+\eta(t),\,\,\,
\label{E1.eq}\\
u(0,x)&=u_0(x),\label{E1.eulic}
\end{align}
where $h\in C^\infty([0,\infty),H^{k+2}_\sigma)$ and $u_0\in
H^k_\sigma$ are given functions, and $\eta$ is the control taking
values in a finite-dimensional subspace $E\subset H^{k+2}_\sigma$.
We denote by $\Theta(h,u_0)$ the set of functions $\eta\in
L^1(J_T,H^k_\sigma)$ for which (\ref{E1.eq}), (\ref{E1.eulic}) has
a unique solution in $C(J_T, H^k_\sigma)$. By Theorem
\ref{T:pert}, $\Theta(h,u_0)$ is an open subset of
$L^1(J_T,H^k_\sigma)$. To simplify the notation, we write
$\RR(\cdot,0,\cdot)=\RR(\cdot,\cdot)$. Let us recall the
definition of  controllability. Suppose $X\subset
L^1(J_T,H^k_\sigma)$ is an arbitrary vector space.
\begin{definition}\label{D.3.1} Eq.~(\ref{E1.eq}) with $\eta \in X$ is said to be controllable
at time $T$ if for any $\e>0$, for any finite-dimensional subspace
$F \subset H^k_\sigma$, for any projection $P_F:H^k_\sigma
\rightarrow H^k_\sigma $ onto $F$ and  for any functions $u_0\in
H^{k}_\sigma$, $\hat{u}\in H^{k}_\sigma$ there is a control $\eta
\in \Theta(h,u_0)\cap X $ such that
\begin{align}
P_F\RR_T(u_0,\eta)&=P_F\hat{u},\nonumber\\
\|\RR_T(u_0,\eta)-\hat{u}\|_{k}&<\e.\nonumber
\end{align}
\end{definition}
 Let us recall some notation introduced in \cite{agr1},
\cite{agr2} and  \cite{shi1}. For any finite-dimensional subspace
$E\subset H_\sigma^{k+2}$, we denote by $\FF(E)$ the largest
vector space $\FF\subset H^{k+2}_\sigma$ such that for any
$\eta_1\in F$ there are vectors $\eta, \zeta^1,\ldots ,\zeta^n\in
E$ and positive constants $ \alpha_1,\ldots,\alpha_n$ satisfying
the relation
\begin{eqnarray}\label{E2:sahmnf}
\eta_1=\eta-\sum_{i=1}^n\alpha_iB(\zeta^i).
\end{eqnarray}
The space $\FF(E)$ is well defined. Indeed, as $E$ is a
finite-dimensional subspace and $B$ is a bilinear operator, then
$\FF(E)$ is contained in a finite-dimensional space.  It is easy
to see that if subspaces $G_1$ and $G_2 $ satisfy
(\ref{E2:sahmnf}), then so does $G_1 + G_2$. Thus, $\FF(E)$ is
well defined. Obviously, $E\subset\FF(E)$. We define $E_k$ by the
rule
\begin{eqnarray}\
E_0=E,\quad E_n=\FF(E_{n-1})\quad \textrm{for} \quad n \geq
1,\quad E_\infty=\bigcup_{n=1}^\infty E_n.\nonumber
\end{eqnarray}
The following theorem is the main result of this section.
\begin{theorem}\label{T.2.1}
Let $h\in C^\infty([0,\infty),H^{k+2}_\sigma)$. If $E\subset
H^{k+2}_\sigma$ is a finite-dimensional subspace such that
$E_\infty$ is dense in $H^k_\sigma$, then  Eq.~(\ref{E1.eq}) with
$\eta\in C^\infty(J_T,E)$ is controllable at any time $T$.
\end{theorem}
\begin{example}
 Let us introduce the functions
\begin{eqnarray}\label{E.cmsmer}
c_m(x) = l(m) \cos\langle m, x \rangle, \,s_m(x) = l(m)
\sin\langle m, x \rangle,
\end{eqnarray}
where  $m \in \Z^3$ and
\begin{eqnarray}\label{E.cmsmer2}
\{l(m),l(-m)\}\quad \text{is an orthonormal basis in
$m^\bot:=\{x\in \R^3, \langle x, m \rangle=0$\}}.
\end{eqnarray}It is shown in \cite {shi1} that if
$$E=span \{c_m,s_m ,\,\,|m| \leq 3\},$$
 then $E_\infty $ is dense in $H^k_\sigma$. We emphasise for what follows that the space $E$ does not depend on the choice of the basis $\{l(m),l(-m)\}$.
\end{example}
The proof of Theorem \ref{T.2.1} is based on the  uniform
approximate controllability of the Euler system.

\begin{definition} Eq.~(\ref{E1.eq}) with $\eta \in X$ is said to be uniformly
approximately controllable at time $T$ if for any  $\e>0$, any
$u_0\in H^{k}_\sigma$  and any compact set $K\subset H^{k}_\sigma$
there is a continuous function $\Psi:K\rightarrow \Theta(h,u_0)
\cap X$ such that
\begin{eqnarray}\label{E2.uniapcon}
\sup_{\hat{u}\in K} \|\RR_T(u_0,
\Psi(\hat{u}))-\hat{u}\|_k<\e,
\end{eqnarray}
where  $ \Theta(h,u_0) \cap X$ is endowed with the  norm of $
L^1(J_T,H^k)$.
\end{definition}
\begin{lemma}\label{R.1}
If for any compact subset $K\subset H^{k+1}_\sigma$  there is a
continuous function $\Psi:K\rightarrow \Theta(h,u_0) \cap X$ such
that (\ref{E2.uniapcon}) holds, then Eq.~(\ref{E1.eq}) with $\eta
\in X$ is
 uniformly approximately controllable at time $T$.
\end{lemma}
\begin{proof}
For any compact set $K\subset H^{k}_\sigma$ there is a small
constant $\delta>0$ such that
\begin{eqnarray}
\sup_{\hat u \in K} \|e^{-\delta L}\hat u- \hat
u\|_k<\frac{\e}{2} \nonumber.
\end{eqnarray}
 As $K_1:=e^{-\delta L}K$ is compact  in  $H^{k+1}_\sigma$, by assumption,  there  is a continuous mapping
$\Psi :K_1 \rightarrow \Theta(h,u_0)\cap C^\infty (J_T,X)$ such
that $$ \sup_{\hat u \in K_1} \|\RR_T(u_0, \Psi(\hat u))-\hat
u\|_k<\frac{\e}{2}.$$ Therefore the continuous mapping
$\Phi:K\rightarrow \Theta(h,u_0)\cap X, $ $\hat u \rightarrow
\Psi(e^{-\delta L}\hat u)$ satisfies the inequality
\begin{eqnarray}
\sup_{\hat u \in K} \|\RR_T(u_0, \Phi (\hat u))-\hat u\|_k<\e.\nonumber
\end{eqnarray}
\end{proof}
The following lemma shows that the uniform approximate
controllability is stronger than controllability.
\begin{lemma}\label{L1:strong}
If Eq.~(\ref{E1.eq}) with $\eta\in X$ is uniformly approximately
controllable at time $T$, then it is also controllable.
\end{lemma}
\begin{proof}
Suppose $F\subset H^k_\sigma$ is a finite-dimensional subspace
 and  $P_F$ is a projection onto $F$, $u_0\in H^{k}_\sigma$ and  $\hat{u}\in
 F$.
Let $B_F(R)$ be the closed ball in $F$ of radius $R$ centred at
origin with $R>M\e$, where $M$ is the norm of $P_F$ and $\e>0$ is
an arbitrary constant. Since $B_F(R)$ is a compact subset of
$H^k_\sigma$, there  is a continuous mapping $\Psi :B_F(R)
\rightarrow \Theta(h,u_0)\cap X$ such that
\begin{eqnarray}\label{E1.astxanish} \sup_{\hat
u \in B_F(R)} \|\RR_T(u_0, \Psi(\hat u))-\hat u\|_k<\e.
\end{eqnarray}
Therefore the continuous mapping $\Phi:B_F(R)\rightarrow F, $
$\hat u \rightarrow P_F\RR_T(u_0, \Psi(\hat u))$ satisfies the
inequality
\begin{eqnarray}
\sup_{\hat u \in B_F(R)} \|\Phi (\hat u)- \hat u\|_k<M\e.\nonumber
\end{eqnarray}
 Fixing $v\in B_F(R-M\e)$  and applying the  Brouwer
theorem to the mapping $u \rightarrow v + u -\Phi(u
):B_F(R)\rightarrow B_F(R)$, we get
\begin{eqnarray}\label{E1.pfcont}
B_F(R-M\e)\subset \Phi(B_F(R)).
\end{eqnarray}
Let $\hat{u}\in F$. By (\ref{E1.pfcont}), for sufficiently large
$R$ there is a function $u_1\in B_F(R) $ such
that\begin{eqnarray}\label{E1.astxanish2}
 P_F\RR_T(u_0, \Psi(u_1))=\hat u.\end{eqnarray}
Using (\ref{E1.astxanish}) and (\ref{E1.astxanish2}), we obtain
\begin{align}
\|\RR_T(u_0,\Psi(u_1))-\hat{u}\|_{k}&\le\|\RR_T(u_0,\Psi(u_1))-u_1\|_k\nonumber\\&+\|u_1-P_F\RR_T(u_0,
\Psi(u_1))\|_{k}< \e+M\e .\nonumber
\end{align}
Since $\e>0$ was arbitrary, this completes the proof.
\end{proof}
Lemma \ref{L1:strong} implies that Theorem \ref{T.2.1} is an
immediate consequence of the following result, which will be
proved in  Sections 5 and 6.
\begin{theorem}\label{T2.himn}
Let $h\in C^\infty([0,\infty),H_\sigma^{k+2})$. If $E\subset
H_\sigma^{k+2}$ is a finite-dimensional subspace such that
$E_\infty$ is dense in $H^k_\sigma$, then  Eq.~(\ref{E1.eq}) with
$\eta\in C^\infty(J_T,E)$ is uniformly approximately controllable
at any time $T$.
\end{theorem}

\section{Controllability of finite-dimensional projections of the velocity and pressure }\label{S:3}

In this section, we are interested in controllability properties
of pressure in Euler system.  We  consider the  problem
(\ref{Pr.E1:1}), (\ref{Pr.E1:2}). If $u\in C(J_T, H^{k})$ is a
solution of (\ref{E1.eq}), (\ref{E1.eulic}), then $(u,p)$ will be
the solution of (\ref{Pr.E1:1}), (\ref{Pr.E1:2}), where
\begin{align}\label{Pr.E4.pitesq}
 p= \Delta^{-1}(\diver h-\sum_{i,j=1}^3 \p_ju_i\p_iu_j).
\end{align}
Here the function $p$ is defined up to the an additive constant
and $\Delta^{-1}$ is the inverse of $\Delta:H^k_\sigma\rightarrow
:H^{k-2}_\sigma$. In what fallows  we normalise $p$ by the
condition that its mean value on $\T^3$ is zero. Denote by
$(\RR(u_0,\eta),\PP(u_0,\eta))$ the solution of (\ref{Pr.E1:1}),
(\ref{Pr.E1:2}) and by $(\RR_t(u_0,\eta),\PP_t(u_0,\eta))$ its
restriction to the time $t$. Eq.~(\ref{Pr.E4.pitesq}) implies that
(\ref{Pr.E1:1}), (\ref{Pr.E1:2})  is not approximately
controllable, so we will be interested in exact  controllability
in projections.

\begin{definition} Eq.~(\ref{Pr.E1:1}) with $\eta \in X$ is said to be exactly controllable in projections
at time $T$ if for any   finite-dimensional subspaces  $F\subset
H^k_\sigma$, $G\subset H^k$ and for any functions $u_0\in
H^k_\sigma$, $\hat{u}\in F$ and $\hat{p}\in G$ there is a control
$\eta \in \Theta(h,u_0)\cap X $ such that
\begin{align}
P_F\RR_T(u_0,\eta)&=\hat{u},\nonumber\\
P_G\PP_T(u_0,\eta)&=\hat{p}.\nonumber
\end{align}
\end{definition}

\begin{theorem}\label{Pr.T1}
If $E\subset H_\sigma^{k+2}$ is a finite-dimensional subspace such
that $E_\infty$ is dense in $H^k_\sigma$, then
Eq.~(\ref{Pr.E1:1}) with $\eta\in C^\infty(J_T,E)$ is exactly
controllable in projections at any time $T>0$.
\end{theorem}
\begin{proof}
To simplify the proof, we shall  assume that $h=0$. The proof
remains literally the same in the case $h\neq 0$. An argument
similar to that used in the proof of Lemma~ \ref{L1:strong} shows
that it suffices to establish the following property:  for any
compact set $K\subset H^k_\sigma\times H^{k}$ and for any constant
$\e>0$ there is a continuous function $\Psi:K\rightarrow
\Theta(h,u_0) \cap X$ such that
\begin{align}
\sup_{(\hat{u},\hat{p} )\in K}\|\RR_T(u_0,
\Psi(\hat{u}, \hat{p}))-\hat{u}\|_k&<\e, \nonumber\\
\sup_{(\hat{u},\hat{p} )\in K}\|P_G\PP_T(u_0, \Psi(\hat{u},
\hat{p}))-\hat{p}\|_k&<\e.\nonumber
\end{align}
We introduce the spaces
\begin{align}
F_m:&=span \{c_n,s_n ,\,\,|n| \leq m, \,\, n\in \Z_*^3 \},\nonumber\\
G_m:&=span\{\sin\langle n, x \rangle, \cos\langle n, x \rangle,
|n|\le m, \,\, n\in \Z^3_*\},\nonumber
\end{align}
where the functions $c_n, s_n$ are defined in (\ref{E.cmsmer}),
(\ref{E.cmsmer2}). By an approximation argument, it suffices to
construct $\Psi$ for any compact set $K\subset F_m \times G_m$.
For an integer $m\ge 1$, we introduce the symmetric quadratic form
$$A( u,v)=-P_{G_m}\Delta ^{-1}\sum_{i,j=1}^3 \p_j u
_i\p_iv_j$$ and set $A(u)=A(u,u).$
 Clearly, we have the following
inequality
\begin{eqnarray}\label{Pr.AA}
 \|A(u)-A(v)\|_k\le C\|u-v\|_k,
\end{eqnarray}
where $u,v \in H_\sigma ^k$ and $C$ is constant depending on
$\|u\|_k+\|v\|_k$. Eq.~(\ref{Pr.E4.pitesq}) implies
\begin{eqnarray}\label{Pr.AA2}
 P_{G_m}\PP_t(u_0,\eta)=A(\RR_t(u_0,\eta)).
\end{eqnarray}
We admit for the moment the following lemma.

\begin{lemma}\label{Pr.L1}  For any $\hat{u}\in F_m$ and $\hat{p}\in G_m$ there is  $v\in
F_m^\bot\cap H^k_\sigma$ such that
\begin{eqnarray}\label{Pr.visah}
\hat{p}=A(\hat u +v),
\end{eqnarray}
where  $F_m^\bot$ is the orthogonal complement of $F_m$ in the
space $H$. Moreover, the  mapping $(\hat u, \hat p)\rightarrow v$ is
continuous from $F_m\times G_m$ to $F_m^\bot$, where $F_m,G_m$ and
$F_m^\bot$ are endowed with the norm of $H^k$.
\end{lemma}

By Theorem \ref{T2.himn}, there is a continuous mapping $\Psi$
such that
\begin{eqnarray}
\sup_{(\hat{u},\hat{p} )\in
K}\|\RR_T(u_0,\Psi(\hat{u},\hat{p}))-(\hat{u}+v)\|_k<\e,\nonumber
\end{eqnarray}
where $v$ satisfies  (\ref{Pr.visah}). From (\ref{Pr.AA}),
(\ref{Pr.AA2})  and  (\ref{Pr.visah}), we have
\begin{align}
\sup_{(\hat{u},\hat{p} )\in K}\|P_{G_m}\PP_T(u_0, \Psi(\hat{u},
\hat{p}))-\hat{p}\|_k&\le \sup_{(\hat{u},\hat{p})\in K} \|A
(\RR_T(u_0,\Psi(\hat{u},\hat{p})))-A (\hat u
+v)\|_k\nonumber\\&\le C\sup_{(\hat{u},\hat{p} )\in
K}\|\RR_T(u_0,\Psi(\hat{u},\hat{p}))-(\hat{u}+v)\|_k. \nonumber
\end{align}
This completes the proof of Theorem \ref{Pr.T1}.

\end{proof}

\begin{proof}[Proof of Lemma \ref{Pr.L1}]
It is easy to see that (\ref{Pr.visah}) is equivalent to
\begin{align}\label{Pr.mech}
&A(v)+2A(\hat u,v)=\hat{p}-A(\hat u)=:\sum_{|n|\le m}
(C_n\sin\langle n, x \rangle +D_n \cos\langle n, x \rangle).
\end{align}

For all $n\in\Z_*^3$, $|n|\le m$ let us take $\{k_n^1\}$,
$\{k_n^2\}$, $\{k_n^3\}$ and $\{k_n^4\}$  in $\Z_*^3$ such that
$|k_n^i|>2m$  and

 \begin{enumerate}
\item[(a)]  $k_n^2-k_n^1=k_n^4-k_n^3=n$,
 \item[(b)] $\min\{|k_n^i+ k_n^j|,|k_n^i\pm k_r^j|,|k_n^3 -k_n^d|,|k_n^4
-k_n^d|  \} >m$,
\item[(c)]$k_n^1$ and  $k_n^3$ are not parallel to  $k_n^2$ and $k_n^4$,
respectively,
\end{enumerate}
for all $i,j= 1,2,3,4$, $d=1,2$,  $|r|<m$ and $ n\neq r$. This
choice is possible. Indeed, let $\phi:\Z_*^3\rightarrow \N_*$ be
an injection and let
\begin{align}\label{E:knisahm}
&k_n^1= 8\phi(n)\mathfrak{m}(n),\qquad\qquad\qquad k_n^3=
(8\phi(n)+4)\mathfrak{m}(n),\nonumber\\& k_n^2=
8\phi(n)\mathfrak{m}(n)+n,\qquad\qquad\, k_n^4=
(8\phi(n)+4)\mathfrak{m}(n)+n,
\end{align}
where $\mathfrak{m}(n)\in \Z^3_*$ is not parallel to $n$ and
$|\mathfrak{m}(n)|=m$. It is easy to see that $\{k_n^j\}$ satisfy
  $(a)-(c)$.      We seek $v$ in the form
$$v=\sum_{|n|\le m} (C_{k_n^1}s_{k_n^1}
+D_{k_n^2}c_{k_n^2}+C_{k_n^3}s_{k_n^3} +C_{k_n^4}s_{k_n^4}).$$
Substituting this expression of $v$ into (\ref{Pr.mech}) and using
the construction of $k_n^i$, we obtain
\begin{align}
&\sum_{|n|\le m} \bigg(A(C_{k_n^1}s_{k_n^1} +D_{k_n^2}c_{k_n^2}) +
A(C_{k_n^3}s_{k_n^3} +C_{k_n^4}s_{k_n^4})
\bigg)\nonumber\\&=\sum_{|n|\le m} (C_n\sin\langle n, x \rangle
+D_n \cos\langle n, x \rangle).\nonumber
 \end{align}
On the other hand,
\begin{align}A(C_{k_n^1}s_{k_n^1} +D_{k_n^2}c_{k_n^2})&=\Delta ^{-1}\sum_{i,j=1}^3
l_i(k_n^1)(k_n^1)_jl_j(k_n^2)(k_n^2)_i
C_{k_n^1}D_{k_n^2}\sin\langle n, x \rangle \nonumber\\&=
-\frac{C_{k_n^1}D_{k_n^2}}{n_1^2+n_2^2+n_3^2} \langle l(k_n^1),
k_n^2 \rangle \langle l(k_n^2), k_n^1 \rangle \sin\langle n, x
\rangle,\nonumber
 \end{align}
 where $l_j(k_n^i)$ and $(k_n^i)_j$ are $j$-th coordinates of $l(k_n^i)$ and $k_n^i $, respectively.
As $k_n^1$ is not parallel to  $k_n^2$, we can choose $l(k_n^1)$
and $l(k_n^2)$   not perpendicular to  $k_n^2$ and  $k_n^1$,
respectively, i.e.,
 $$\langle l(k_n^1),
k_n^2 \rangle \langle l(k_n^2), k_n^1 \rangle  \neq 0.$$ Hence,
there are  constants $C_{k_n^1}, D_{k_n^2}$ continuously depending
on $C_n$, and therefore on $(\hat u, \hat p)$, such that
$$A(C_{k_n^1}s_{k_n^1}
 +D_{k_n^2}c_{k_n^2})=C_n \sin\langle n, x\rangle
. $$ In the same way, we can choose $C_{k_n^3}, C_{k_n^4}$ such
that
\begin{align}
A(C_{k_n^3}s_{k_n^3} +C_{k_n^4}s_{k_n^4})= D_n \cos\langle n, x
\rangle.\nonumber
 \end{align}
 Thus we have (\ref{Pr.visah}).

\end{proof}

\section{Proof of Theorem \ref{T2.himn}}\label{S:4}

Let us fix a constant $\e>0$, an initial point $u_0\in
H^{k}_\sigma$,
 a compact set $K\subset H^{k}_\sigma$ and a vector subspace  $X\subset L^1(J_T,H_\sigma^k)$. Eq.~(\ref{E1.eq}) with
$\eta\in X$ is said to be uniformly $(\e,u_0,K)$-controllable at
  time $T>0$ if there is a continuous mapping $$\Psi:K\rightarrow \Theta(h,u_0)\cap X$$
such that
\begin{eqnarray}\label{E2.unif}
\sup_{\hat{u}\in K}
\|\RR_T(u_0,\Psi(\hat{u}))-\hat{u}\|_k<\e,\nonumber
\end{eqnarray}
where  $ \Theta(h,u_0) \cap X$ is endowed with the   norm of $
L^1(J_T,H^k_\sigma)$.

Theorem \ref{T2.himn} is deduced from the following result, which
is established in  next section.

\begin{theorem}\label{T.reduct} Let  $E\subset
H_\sigma^{k+2}$ be a finite-dimensional subspace. If
Eq.~(\ref{E1.eq}) with $\eta \in C^\infty(J_T,\FF(E))$ is
uniformly $(\e,u_0,K)$-con\-trollable, then it is also
$(\e,u_0,K)$-controllable
 with $\eta\in C^\infty(J_T,E)$.
\end{theorem}

\begin{proof}[Proof of Theorem \ref{T2.himn}]

We first prove that there is an integer $N\ge 1$ depending only on
$\e $, $u_0 $ and $K $ such that Eq.~(\ref{E1.eq}) with $\eta\in
C(J_T,E_N)$ is uniformly $(\e,u_0,K)$-controllable  at time $T$.
Let us define a continuous operator defined on $J_T\times K$ by
$$u_{\mu,\delta}(t, \hat u) = T^{-1}(te^{\mu L}\hat u + (T - t)e^{\delta L}u_0)$$
It is easy to see that $u_{\mu,\delta}$ satisfies
Eq.~(\ref{E1.eq}) with
$$\eta_{\mu,\delta} =\dot{u}_{\mu,\delta}+B (u_{\mu,\delta})-h(t).$$
As $K$ is a compact set in $H^k_\sigma$, we have
\begin{align}\sup_{\hat u \in K}\|u_{\mu,\delta}(T, \hat u)-\hat u\|_k&\rightarrow 0 \,\, \text{as}\,\,\mu \rightarrow 0\nonumber,\\
\sup_{\hat u \in K} \|u_{\mu,\delta}(0, \hat u)- u_0\|_k&\rightarrow 0 \,\, \text{as}\,\, \delta \rightarrow 0. \nonumber
 \end{align}
The fact that $E_\infty$ is dense in $H^k_\sigma$ implies
$$ \|P_{E_N} \eta_{\mu,\delta} - \eta_{\mu,\delta}\|_{L^1(J_T,H^{k})} \rightarrow 0 \,\, \text{as}\,\,N \rightarrow \infty.$$
By Theorem \ref{T:pert}, we can chose  $N$, $\mu$ and $\delta$ such that
$$ \sup _{\hat u \in K}\|\RR (u_0, P_{E_N}\eta_{\mu,\delta}(\hat u))-\hat u \|_k<\e.$$
We note that   the mapping $P_{E_N}\eta_{\mu,\delta}
(\cdot,\cdot):\hat{u}\rightarrow
P_{E_N}\eta_{\mu,\delta}(\cdot,\hat{u})$   is continuous from $K$
to $C(J_T,H^k_\sigma)$. Hence Eq.~(\ref{E1.eq}) is uniformly
$(\e,u_0,K)$-controllable with $\eta\in C(J_T,E_N)$. Applying $N$
times Theorem \ref{T.reduct}, we complete the proof of Theorem
\ref{T2.himn}.
\end{proof}

\section{Proof of Theorem \ref{T.reduct}}\label{S:5}
The proof of Theorem \ref{T.reduct} is inspired by ideas from
\cite{agr1, agr2, shi1,shi2}. Let us consider  the following
control system:
\begin{align}
\dot{u}+ B(u+\zeta)&=h+\eta,\label{E2.conv}
\end{align}
where $\eta, \zeta$ are $E$-valued controls.
 Let $\hat{\Theta}(u_0,h)$ be the set of pairs
$(\eta, \zeta)\in L^1(J_T,H_\sigma^k)\times
L^2(J_T,H_\sigma^{k+1})$ for which problem (\ref{E2.conv}),
(\ref{E1.eulic}) has a unique solution in $C(J_T,H_\sigma^{k})$.
Eq (\ref{E2.conv}) with $(\eta, \zeta)\in \hat X\subset
L^1(J_T,H_\sigma^k)\times L^2(J_T,H_\sigma^{k+1})$ is said to be
uniformly $(\e,u_0,K)$-controllable if there is a continuous
mapping
$$\hat{\Psi}:K\rightarrow  \hat{\Theta}(h,u_0)\cap \hat X $$
such that
\begin{eqnarray}\label{E2.unif1}
\sup_{\hat{u}\in K} \|\RR_T(u_0,
\hat{\Psi}(\hat{u}))-\hat{u}\|_k<\e,
\end{eqnarray}
where $\hat{\Theta}(h,u_0)\cap \hat X$ is endowed with the norm of
$ L^1(J_T,H_\sigma^k)\times L^2(J_T,H_\sigma^{k+1})$.

We claim that, when proving Theorem \ref{T.reduct}, it suffices to
assume   $u_0\in H_\sigma^{k+2}$. Suppose that for any $v_0\in
H_\sigma^{k+2}$ and for any continuous mapping $\Phi:K \rightarrow
{\Theta}(h,v_0)\cap C^\infty (J_T,E_1)$ there is a continuous
mapping
$$\hat{\Phi}:K\rightarrow  {\Theta}(h,v_0)\cap C^\infty (J_T,E)$$
such that
\begin{eqnarray}
\sup_{\hat{u}\in K}\|\RR_T(v_0,\Phi(\hat u))-\RR_T(v_0,\hat\Phi(\hat u))\|_{k}<\frac{\e}{3}.\nonumber
\end{eqnarray}
Let us show that for any $u_0\in H_\sigma^k$  and for any
continuous mapping   $\Psi:K \rightarrow {\Theta}(h,v_0)\cap
C^\infty (J_T,E_1)$   there is a continuous mapping
$\hat{\Phi}:K\rightarrow  {\Theta}(h,v_0)\cap C^\infty (J_T,E)$
such that
$$\sup_{\hat{u}\in K} \|\RR_T(u_0,
\Psi(\hat{u}))-\RR_T(u_0, \hat \Psi(\hat{u}))\|_k<\e.$$
 By Theorem \ref{T:pert}, there is $v_0\in H_\sigma^{k+2}$ such that
\begin{equation}\label{E3:1 }
  \sup_{\hat{u}\in K}  \|  \RR(u_0,\Psi(\hat u))-\RR(v_0,\Psi(\hat u))    \|_{C(J_T,H^{k})}< \frac{\e}{3}.
\end{equation}
By our assumption, as $v_0 \in H_\sigma^{k+2}$, there is a
continuous mapping $$\hat{\Psi}_{\e}:K\rightarrow
{\Theta}(h,u_0)\cap C^\infty (J_T,E) $$ such that
\begin{eqnarray}\label{E3.un11}
\sup_{\hat{u}\in K}\|\RR_T(v_0,\Psi(\hat
u))-\RR_T(v_0,\hat\Psi_{\e}(\hat u))\|_{k}<\frac{\e}{3}.
\end{eqnarray}
 By Theorem \ref{T:pert}, we have
 \begin{equation}\label{E3:3 }
 \|\RR_T(v_0 ,\hat\Psi_{\e}(\hat u))-\RR_T(u_0,\hat\Psi_{\e}(\hat
 u))\|_{k}\le C\|v_0-u_0\|,
\end{equation}
where $C$ is a constant not depending  on $\e $. Choosing $v_0 $
sufficiently close to $u_0$ and using  inequalities (\ref{E3:1 }),
(\ref{E3.un11}) and (\ref{E3:3 }), we get
$$\| \RR_T(u_0,\Psi(\hat u)) -\RR_T(u_0,\hat\Psi_{\e}(\hat u)) \|_{k}< \e.$$
 From now on, we assume that $u_0\in H_\sigma^{k+2}$. In this case,  Theorem \ref{T.reduct}  is deduced  from the following two propositions.
\begin{proposition}\label{P.1}
Eq.~(\ref{E1.eq}) with $\eta \in C^\infty(J_T,E)$ is uniformly
$(\e,u_0,K)$-con\-trollable if and only if so is
Eq.~(\ref{E2.conv}) with $(\eta, \zeta)\in C^\infty(J_T,E\times
E)$.
\end{proposition}
\begin{proposition}\label{P.2}
Eq.~(\ref{E2.conv}) with $(\eta, \zeta)\in C^\infty(J_T,E\times
E)$ is uniformly $(\e,u_0,K)$-controllable if and only if so is
Eq.~(\ref{E1.eq}) with $\eta_1\in C^\infty (J_T,E_1)$.
\end{proposition}
\begin{proof}[Proof of Proposition \ref{P.1}]$ $
We show that if  (\ref{E2.conv}) with $(\eta, \zeta)\in
C^\infty(J_T,E\times E)$ is uniformly $(\e,u_0,K)$-controllable,
then so is   (\ref{E1.eq}) with $\eta \in C^\infty(J_T,E)$.
 Let $$\hat{\Psi}:K\rightarrow  \hat{\Theta}(h,u_0) \cap C^\infty(J_T,E\times E), \qquad
\hat{\Psi}(\hat{u})=\big(\eta(t,\hat{u}),\zeta(t,\hat{u})\big)$$
be such that
\begin{eqnarray} \label{E5.epsglx}
\hat\e:=\sup_{\hat{u}\in K} \|\RR_T(u_0,
\hat{\Psi}(\hat{u}))-\hat{u}\|_k<\e.
\end{eqnarray}
 Let us choose  $\zeta_n(\cdot,\hat u)\in C^\infty(J_T,E)$  such that $\zeta_n(0)=\zeta_n(T)=0$, the mapping $\zeta_n(\cdot,\cdot):\hat u
  \rightarrow  \zeta_n(\cdot,\hat u)$ from $K$ to $C^1(J_T,H_\sigma^{k+1})$ is continuous  and
$$\|\zeta_n-\zeta\|_{L^2(J_T,H^{k+1})} \rightarrow 0 \,\, \text {as} \,\,
n\rightarrow \infty.$$ By Theorem \ref{T:pert}, for sufficiently
large $n$ we have
\begin{eqnarray}\label{E4.5.1}
\sup_{\hat u \in K}\|
\RR_T(u_0,\zeta_n(\hat{u}),\eta)-\RR_T(u_0,\hat{\Psi}(\hat u))
\|_k< \e -\hat\e.
\end{eqnarray}
Define  $\Psi_n(t,\hat u)=\eta(t,\hat{u})-\dot
{\zeta}_n(t,\hat{u}) $. It is easy to see that
$\Psi_n(\cdot,\cdot):\hat{u}\rightarrow \Psi_n(\cdot,\hat{u})$ is
a continuous mapping from $K$ to $L^1(J_T,H_\sigma^k)$. Clearly,
$$\RR (u_0,\zeta_n(\hat{u}),\eta)=\RR (u_0,\Psi_n(\hat{u}))-\zeta_n(\hat{u}).$$
Using the fact that $\zeta_n(T)=0$, (\ref{E4.5.1}) and
(\ref{E5.epsglx}), we derive
$$\sup_{\hat u \in K}\| \RR_T(u_0,\Psi_n(\hat{u}),\eta)-\hat u \|_k<  \e -\hat\e+ \sup_{\hat u \in K}\| \RR_T(u_0,\hat{\Psi}(\hat u))-\hat u \|_k< \e, $$
which completes the proof of Proposition \ref{P.1}.
\end{proof}

\begin{proof}[Proof of Proposition \ref{P.2}]$ $
By Proposition \ref{P.1} and   the fact $E\subset E_1$, if
Eq.~(\ref{E2.conv}) is uniformly $(\e,u_0,K)$-controllable, then
so is  Eq.~(\ref{E1.eq}) with $\eta\in C^\infty (J_T,E_1)$. We
need to prove the converse assertion. We assume that there  is a
continuous mapping
$$\Psi_1:K\rightarrow  {\Theta}(h,u_0)\cap L^1(J_T,E_1) $$  such that
\begin{eqnarray}
\hat\e:=\sup_{\hat{u}\in K} \|\RR_T(u_0,
\Psi_1(\hat{u}))-\hat{u}\|_k<\e.\nonumber
\end{eqnarray}
We approximate $\RR_T(u_0, \Psi_1(\hat{u}))$ by a solution
$u(t,\hat{u})$ of problem (\ref{E2.conv}), (\ref{E1.eulic}) with
some $\eta(t,\hat{u})$, $\zeta(t,\hat{u})\in C^\infty( J_T, E )$
such that $(\eta(t,\hat{u}), \zeta(t,\hat{u})) $
 depends continuously on $\hat{u}\in K$.

  \vspace{6pt}  \textbf{Step 1.}
We first approximate $ \Psi_1(\hat{u})$ by a family of piecewise
constant controls. Let us introduce a finite set $A = \{\eta^l_1
\in E_1 , l = 1, \ldots ,m\}$. For any integer $s$, we denote by
$P_s(J_T ,A)$ the set of functions

\begin{align}
\eta_1(t)=\sum_{l=1}^m\ph_l(t)\eta_1^l \,\, \text { for }\,\,
t\in[0,T],\nonumber
\end{align}
where  $\ph_l$ are non-negative functions such that
$\sum_{l=1}^m\ph_l(t)=1$, $$\ph_l(t)=\sum_{r=0}^{
s-1}c_{l,r}I_{r,s}(t)\,\, \text { for }\,\,t\in[0,T],$$and
$I_{r,s}$ is the indicator function of the interval
$[t_r,t_{r+1})$ with $t_r=rT/s$.

We define a metric in $P_s(J_T ,A)$ by
$$ d_P(\eta_1,\zeta_1)=\sum_{l=1}^m \| \ph_l-\psi_l\|_{L^{\infty}(J_T)}, \,\, \eta_1,\zeta_1\in P_s(J_T ,A),$$
where $\{\ph_l\}$ and $\{\psi_l\}$ are the functions corresponding
 to $\eta_1$ and $\zeta_1$, respectively. We shall need
the following lemmas, which are proved at the end of this section.

\begin {lemma}\label{L3.1}
If Eq.~(\ref{E1.eq}) with $\eta\in C^\infty (J_T,E_1)$ is
uniformly $(\e,u_0,K)$-controllable, then there is a finite set $A
= \{\eta_1^l , l = 1,  \ldots ,m\}   \subset E_1$, an integer
$s\ge1$ and a  mapping $\Psi_s:K\rightarrow P_s(J_T ,A)$
continuous with respect to the metric of $P_s(J_T ,A)$ such that
$\Psi_s(K)  \subset \Theta (u_0,h)$ and
\begin{eqnarray}
\sup_{\hat{u}\in K} \|\RR_T(u_0,
\Psi_s(\hat{u}))-\hat{u}\|_k<\e.\nonumber
\end{eqnarray}
\end {lemma}
\begin {lemma}\label{L:L}
 Let $E \subset H_\sigma^{k+2}$ be a finite-dimensional space and $E_1 = \FF(E)$. Then
for any $\eta_1 \in E_1$ there are vectors $\zeta^1, \ldots ,
\zeta^p, \eta \in E$ and positive constants $\lambda_1, \ldots ,
\lambda_p$  whose sum is equal to 1 such that
\begin{equation}
B(u)-\eta_1= \sum_{j=1}^p \lambda_j B(u+\zeta^j) -\eta\quad
\text{for any u }\in H^1.\nonumber
\end{equation}
 \end {lemma}
Let $\Psi_s$ be the function constructed in Lemma \ref{L3.1}:
$$ \Psi_s(\hat{u})=\sum_{l=1}^m\ph_l(t,\hat u)\eta_1^l.$$

As $\eta_1^l\in E_1$, by Lemma \ref{L:L}, there are vectors
$\zeta^{l,1}, \ldots , \zeta^{l,p}, \eta^{l} \in E$ and positive
constants $\lambda_{l,1}, \ldots , \lambda_{l,p}$  whose sum is
equal to 1 such that
\begin{equation}\label{E.lemma1}
B(u)-\eta_1^l= \sum_{j=1}^p \lambda_{l,j} B(u+\zeta^{l,j})
-\eta^{l}\quad \text{for any $u$ }\in H^1.
\end{equation}
Let $u_1=\RR(u_0,\Psi_s(\hat{u}))$. It follow from
(\ref{E.lemma1})  that $u_1$ satisfies the equation
\begin{eqnarray}\label{E3.aftleem1}
\dot{u}_1+\sum_{j=1}^p\sum_{l=1}^m \lambda_{l,j}\ph_l(t,\hat{u})
B(u_1+\zeta^{l,j}) =h(t)+\sum_{l=1}^m\ph_l(t,\hat{u})\eta^l.
\end{eqnarray}
We can rewrite Eq.~(\ref{E3.aftleem1}) in the form
\begin{eqnarray}\label{E3.aftleem2}
\dot{u}_1+\sum_{i=1}^q \psi_i(t,\hat{u}) B(u_1+\zeta^{i})
=h(t)+\eta (t,\hat{u}),
\end{eqnarray}
where $\zeta^{i}\in E$ for $ i=1,\ldots,q$,
$\eta(t,\hat{u})=\sum_{l=1}^m\ph_l(t,\hat{u})\eta^l$
 such that
\begin{eqnarray}
\psi_i(t,\hat{u})=\sum_{r=0}^{s-1}
d_{i,r}(\hat{u})I_{r,s}(t),\,\,\,\sum_ {i=1}^qd_{i,r}=1\nonumber
\end{eqnarray}
for some non-negative functions $d_{i,r}\in C(K)$.

  \vspace{6pt} \textbf{Step 2.} We  approximate $u_{1}$ by a solution
  of problem (\ref{E2.conv}), (\ref{E1.eulic}). First we assume $s=1$. In this case
(\ref{E3.aftleem2}) becomes
\begin{eqnarray}\label{E3.case1}
\dot{u}_1+\sum_{i=1}^q d_i(\hat{u}) B(u_1+\zeta^{i}) =h(t)+\eta
(\hat{u}),
\end{eqnarray}
where $d_i\in C(K)$ and $\eta\in C(K,E) $.  Let
$\zeta_n(t,\hat{u})=\zeta(\frac {nt}{T},\hat{u})$, where
$\zeta(t,\hat{u})$  is a $1$-periodic function such that
$$\zeta(s,\hat{u})=\zeta^j\text{ for }0\le
s-(d_1(\hat{u})+\ldots+d_{j-1}(\hat{u}))<d_j(\hat{u}),\quad
j=1,\ldots ,q,$$ where $d_0(\hat{u})=0$. Eq.~(\ref{E3.case1}) is
equivalent to the  equation
\begin{equation}
\dot u_1   + B(u_1 + \zeta_n(t,\hat{u}))  = h(t) + \eta(t,\hat{u})
+ f_n(t,\hat{u}),\nonumber
\end{equation}
where
\begin{equation}\label{E3.fnitesq}
f_n(t,\hat{u})= B(u_1 + \zeta_n(t,\hat{u})) -\sum_{i=1}^q d_i(\hat{u}) B(u_1+\zeta^{i}) .
\end{equation}
 Let us define
$$\KK g(t)=\int_0^t  g(s)\dd s.$$
Then $v_n = u_1 - \KK f_n $ is a solution of the problem
\begin{align}
\dot v_n +B (v_n + \zeta_n(t,\hat{u})+ \KK f_n(t,\hat{u})) & = h(t)
+ \eta(t,\hat{u}),\nonumber\\
v_n&=u_0.\nonumber
 \end{align}
Suppose we have shown that
\begin{equation}\label{E3.fk}
 \sup_{\hat{u}\in K} \|\KK f_n(t,\hat{u})\|_{C(J_T,H^{k+1})}\rightarrow 0.
\end{equation}
Then  $v_n$ satisfies
\begin{equation}
  \|v_n-u_1\|_{C(J_T,H^{k+1})}\rightarrow 0 \,\, \text{as} \,\,
  n\rightarrow \infty.\nonumber
\end{equation}
There is an integer $n_0\ge 1$ such that if $n\ge n_0$
\begin{equation}
 \sup_{\hat{u}\in K}\|\RR(u_0,\zeta_n(\hat{u}),\eta(\hat{u}))-u_1(\cdot,\hat u)\|_{C(J_T,H^k)}<\e-\hat\e.\nonumber
\end{equation}
 Then the operator
$$\hat{\Psi}_n:K \rightarrow L^1(J_T,E)\times L^2(J_T,E),\,\,\, \hat{u}\rightarrow (\eta(\hat{u}),\zeta_n(\hat{u}))$$
satisfies (\ref{E2.unif1}).

 To finish the proof of Proposition \ref{P.2} in
 the case $s=1$, it suffices to prove (\ref{E3.fk}). Suppose  we have
shown that
\begin{equation}\label{E3.fk2}
 \|\KK f_n(t,\hat{u})\|_{C(J_T,H^{k+1})}\rightarrow 0 \,\,\, \text{for any } \, \hat{u}\in
 K.
\end{equation}
To prove (\ref{E3.fk}), by the Arzel\`{a}--Ascoli theorem, it
suffices to show that the family $\{\hat{u}\rightarrow \KK
f_n(\cdot,\hat{u})\}$ is uniformly equicontinuous from $K$ to
$C(J_T, H_\sigma^{k+1})$. By (\ref{E3.fnitesq}), it suffices to
show that so is $\hat{u}\rightarrow \zeta_n(\hat{u})$ from $K$ to
$L^1(J_T, H_\sigma^{k+2})$. The definition of $\zeta_n $ implies
\begin{align}
 &\|\zeta_n(\cdot,\hat{u}_1)-\zeta_n(\cdot,\hat{u}_2)\|^2_{L^2(J_T,H^{k+2})}\le
 \int_0^T\| \zeta (\frac {nt}{T},\hat{u}_1)-\zeta(\frac
 {nt}{T},\hat{u}_2)\|_{{k+2}}^2\dd t\nonumber\\&= \frac{T}{n} \int_0^n\| \zeta
(t,\hat{u}_1)-\zeta(t,\hat{u}_2)\|_{{k+2}}^2\dd
 t\le C\sum_{i=1}^q|d_i(\hat{u}_1)-d_i(\hat{u}_2)|.\nonumber
\end{align}
The uniform  continuity of $d_i$ over $K$ gives us the required
result.

 \vspace{6pt} \textbf{Step 3.} To complete the proof of  Proposition \ref{P.2} in the case $s=1$, it remains to prove
 (\ref{E3.fk2}).  If we show that for any piecewise constant $H_\sigma^{k+2}$-valued function $u_1$
on $J_T$, the sequence $\{\KK f_n\}$ converges to zero in the
space $C(J_T ,H_\sigma^{k+1})$, then  an approximation argument
shows (\ref{E3.fk2}) for any $u_1 \in C(J_T,H_\sigma^{k+2})$.

The family $\{\KK f_n\}$ is relatively compact in the space $C(J_T
,H_\sigma^{k+1})$ for any piecewise constant function $u_1$.
Indeed, the set $f_n(t), t \in J_T$ is contained in a finite
subset of $H_\sigma^{k+1}$ not depending on $n$. Thus, there is a
compact set $G \subset H_\sigma^{k+1}$ such that
$$ \KK f_n(t) \in G \text{ for all }t \in J_T ,  n \ge 1.$$
As
\begin{equation}
\sup_{n\ge1}\| f_n\|_{C(J_T ,H^{k+1})} < \infty,\nonumber
\end{equation}
the family $\{\KK f_n\}$ is uniformly equicontinuous on $J_T$.
Thus, by the Arzel\`{a}--Ascoli theorem, $\{\KK f_n\}$ is
relatively compact. Therefore  convergence (\ref{E3.fk2})  will be
established if we show that
\begin{equation}\label{E3:prc}
\KK f_n(t) \rightarrow 0 \text{ in }  H_\sigma^{k+1} \text{ for
any } t \in J_T .
\end{equation}
To prove (\ref{E3:prc}), we first assume that $u(t)=b\in
H_\sigma^{k+2}$ for all $t\in J_T$. Let $t=t_l+\tau$, where
$t_l=\frac{lT}{n}$, $l\in \N$ and $\tau\in [0,\frac{T}{n})$. From
the definition of $\zeta_n(t)$ we have
\begin{eqnarray}
\int_0^{\frac{lT}{n}}f_n(s)\dd s= \int_0^{\frac{lT}{n}}\bigg(B(b +
\zeta_n(t)) \bigg)\dd s-\frac{lT}{n}\sum_{j=1}^p \lambda_jB
(b+\zeta^j) =0, \nonumber
\end{eqnarray}
so
\begin{equation}
\KK f_n(t)=-\tau \sum_{j=1}^p \lambda_jB (b+\zeta^j)  - \int_0^\tau
B(b + \zeta_n(s))\dd s.\nonumber
\end{equation}
Since  $\tau \rightarrow 0 $ as $n \rightarrow \infty $, we arrive
at (\ref{E3:prc}). In the same way, we can show that
(\ref{E3:prc}) holds for any piecewise constant function $u$.

The case $s\ge2$ is deduced from the case $s=1$ exactly in the
same way as in \cite[section 3.3]{shi2}.
\end{proof}
\begin{proof}[Proof of Lemma \ref{L3.1}]$ $
Let $\{e_1,\ldots,e_d\}$ be an orthonormal basis in $E_1$ with
respect to scalar product $\langle \cdot,\cdot\rangle$  and
$\xi_l(t, \hat u):=\langle\Psi_1(t,\hat u), e_l\rangle $ for
$l=1,\ldots,d$.  Let us define for $M=\max_{l,t,\hat
u}|\xi_l(t,\hat u)|$ and $m=2d$ the vectors
$$\eta_1^l=dMe_l\,\,\text{for $l=1,\ldots,d$,} \,\, \eta_1^l=-dMe_l\,\,\text{for $l=d+1,\ldots,m$. }$$
We can see that the functions $$ \tilde{\xi}_l(t,\hat
u)={\frac{1}{2d}\big(1+\frac{\xi_l(t,\hat
u)}{M}}\big),\,\,\tilde{\xi}_{l+d}(t,\hat
u)={\frac{1}{2d}\big(1-\frac{\xi_l(t,\hat u)}{M}}\big)\,\,
\text{for $l=1,\ldots,d$  } $$ are non-negative, their sum is
equal to 1, and they satisfy the relation
$$\Psi_1(t,\hat u)=\sum _{l=1}^m\tilde{\xi}_{l }(t,\hat
u)\eta_1^l.  $$
Let us define an operator $\Psi_s :K \rightarrow
P_s(J_T ,A)$ with $A=\{\eta_1^l, l=1,\ldots,m \}$ as
$$\Psi_s(t,\hat u)=\sum_{l=1}^m\tilde{\xi}_{l }(\frac{rT}{s},\hat
u)\eta_1^l\,\,\text{for} \,\, t\in
[\frac{rT}{s},\frac{(r+1)T}{s}).$$
 Since $\tilde{\xi}_l \in C(J_T \times K)$ and $K\subset H_\sigma^{k+2}$ is
 compact, we have
$$ \sup_{\hat u \in K}\|\Psi_1(t,\hat u)-\Psi_s (t,\hat u)\|_{k+2}=\sup_{\hat u \in K}\|\sum _{l=1}^m\big(\tilde{\xi}_{l }(t,\hat
u)-\tilde{\xi}_{l }(\frac{rT}{s},\hat
u)\big)\eta_1^l\|_{k+2}\rightarrow0 \,\,\text{as}\,\, s\rightarrow
\infty.$$ Thus for sufficiently large $s$, we have $\Psi_s(K)\subset
\Theta(h,u_0)$, and $$\sup_{\hat u \in K}\|\RR_T(u_0,\Psi_s(\hat
u))-\RR_T(u_0,\Psi_1(\hat u))\|_k< \e.$$ Hence (\ref{E2.conv}) is
uniformly $(\e,u_0,K)$-controllable with $\eta\in P_s(J_T,A).$

\end{proof}

\begin {proof}[Proof of Lemma \ref{L:L}]
By the definition of $\FF(E)$, for any $\eta_1\in \FF(E)$ there
are $\xi^1,\ldots ,\xi^n, \eta \in E$ and positive constants $
\alpha_1,\ldots ,\alpha_n$ such that $$
\eta_1=\eta-\sum_{i=1}^n\alpha_iB(\xi^i).$$ Let us set $p = 2n,$
$\alpha=\alpha_1+\ldots+\alpha_n$,
$$
\lambda_i = \lambda_{i+n} = \frac {\alpha_i}{2\alpha},
\quad\zeta^i = -\zeta^{i+n} = \sqrt{\alpha}\xi^i, \quad i = 1,
\ldots ,n.
$$
Then we have
\begin{equation}
B(u) -\eta_1= \sum_{j=1}^p \lambda_jB(u+\zeta^j) -\eta\quad
\text{for any $u$ }\in H_\sigma^1.\nonumber
\end{equation}
\end{proof}

\section{Non controllability result}\label{S:6}
 Let us denote by $A_T(u_0,h,E)$ the set of
attainability at time $T$ from $u_0\in H_\sigma^k$ by $E$-valued
controls, i.e., $$A_T(u_0,h,E)=\{\hat u\in H_\sigma^k:\,\,\hat
u=\RR_T(u_0,\eta) \,\, \text{for some} \,\, \eta\in \Theta(u_0,h)
\}.$$   In this section, we show that the ideas of \cite{shi3} can
be generalized
 to prove that the set $A(u_0,h,E)=\cup_{T\in [0,\infty)}
A_T(u_0,h,E)$   does not contain a ball of $H_\sigma^{k+\gamma}$,
$\gamma<2$ in the three-dimensional  case.

 Let us recall the definition of Kolmogorov
$\e$-entropy (see \cite{lor}). For any $\e > 0$, we denote by
$N_\e(K)$ the minimal number of sets of diameters not exceeding
$2\e$ that are needed to cover $K$. The Kolmogorov $\e$-entropy of
$K$ is defined as $H_\e(K) = \ln N_\e(K)$.

 Let us
consider the equation
\begin{equation}\label{E4.hamar}
\dot{v}+ B(v+z)=h.
\end{equation}
We fix an integer $k\ge4$ and denote by $\Theta_{t}(h,u_0)$ the
set of functions $\eta\in L^1(J_t,H_\sigma^k)$ for which
(\ref{E4.hamar}), (\ref{E1.eulic}) with $z(t)=\int_0^t\eta(s)\dd
s$ has a unique solution $v\in C(J_t, H_\sigma^k)$. We note that
$$\RR_t(u_0,\eta)=v(t)+z(t),$$
where $z(t)=\int_0^t\eta(s)\dd s$. The following theorem is the
main result of this section.
\begin{theorem}
 Let $k\ge 4$, $u_0\in H_\sigma^k$, $h\in C([0,\infty), H_\sigma^k)$ and $E\subset H_\sigma^k$ be any finite-dimensional subspace.
For  any   $\gamma \in [0,2)$ and any ball $Q\subset
H_\sigma^{k+\gamma}$, we have
\begin{equation}
 A^c(u_0,h,E)\cap Q \neq \emptyset,\nonumber
\end{equation}
where $A^c(u_0,h,E)$ is the complement of $A(u_0,h,E)$ in the
space $H_\sigma^{k}$.
\end{theorem}
\begin{proof}We argue by contradiction.
 Suppose that $ A(u_0,h,E)$ contains a closed ball $Q\subset H_\sigma^{k+\gamma}$.
  Let   $\{t_l\}$ be a dense sequence in $[0,\infty)$ and let
\begin{align}
&D_{l,n}:=\{(z,y)\in W^{1,1}(J_{t_l},H_\sigma^k)\cap
\Theta_{t_l}(u_0,h)\times E :\,
 \|z\|_{W^{1,1}(J_{t_l},H^k)}\le n, \|y\|_k\le n\},\nonumber\\
&B_{l,n}:=\{ \hat u \in H_\sigma^k  :\, \hat u = \RR_t(u_0,z,h)+y
\,\,\text{for some} \,\, (z,y)\in D_{l,n} , t\in[0,t_l]
\}.\nonumber
\end{align}
It is easy to see that $\bigcup_{l,n} B_{l,n}\supset A(u_0,h,E)$. By the
Baire theorem, there  are integers $p$ and $m$ such that $B_{p,m}$
is dense in a ball $\hat Q$ with respect to the metric of
$H_\sigma^{k+\gamma}$. Let us denote by $K:[0,\infty)\times
L^1([0,\infty),E) \times E\rightarrow H_\sigma^{k-1}$ the
continuous operator that takes the triple $(t,z,y)\in
J_{t_p}\times D_{p,m}$ to $\RR_t(u_0,z,h)+y$. As $K(J_{t_p}\times
D_{p,m})\subset B_{p,m}$ is closed in $H_\sigma^{k+\gamma}\cap
B_{p,m}$, then $\hat Q \subset B_{p,m} $. We have from \cite {edm}
\begin{equation}\label{E4.pr2}
H_\e(Q,L^2)\sim \bigg( \frac{1}{\e}\bigg)^{\frac{3}{k}},
\end{equation}
where $Q$ is a ball in $ H^k$. To obtain (\ref{E4.pr2}) for any
$Q\subset H_\sigma^{k}$, we follow the ideas of  \cite
[Proposition 2.2]{shi3}. Let us denote by $\Sigma^k$ the closure
in $H^k$ of the set of functions  $u=(\p_2v,-\p_1v,0)\in H^{k}$,
where $v \in H^{k+1}$ is a scalar function. Since $\Sigma^k$ is
closed subspace of $H_\sigma^{k}$, it suffices to prove
(\ref{E4.pr2}) any  ball $Q\subset \Sigma^k$. Let us introduce the
set of scalar functions
$$\dot{H}^k(\T^3):=\{u\in H^{k}(\T^3): \int_{0 }^{2\pi}u(x_1,x ') \dd x_1=0 \text{ for any $x'\in \T^2$} \}.$$
As
$$H^k(\T^3)=\dot{H}^k(\T^3)\dot + H^k(\T^2),$$
 we get  (\ref{E4.pr2}) for any ball $Q\subset \dot{H}^k(\T^3)$.
Finally, if $\Pi_2$ is the projection
$\Pi_2(u_1,u_2,u_3)\rightarrow u_2$, then $\Pi_2\Sigma^k=
\dot{H}^k(\T^3)$. Thus  (\ref{E4.pr2}) holds for any $Q\subset
H_\sigma^{k}$. Hence
\begin{equation}\label{E4.pr21}
H_\e( Q,H^{k-1})\sim \bigg( \frac{1}{\e}\bigg)^{\alpha},
\end{equation}
where $Q$ is a ball in $H_\sigma^{k+\gamma}$ and
$\alpha=\frac{3}{1+\gamma}>1$. On the other hand,  from
\cite[(3.10)]{shi3} it follows that
\begin{equation}\label{E4.pr3}  H_\e\bigg(J_{t_p}\times
D_{p,m},\R\times  L^1(J,E)\times E \bigg)\prec
\frac{1}{\e}\ln\frac{1}{\e}.
\end{equation}
As $h\in C([0,\infty), H_\sigma^k)$, by Theorem \ref{T:pert}, the
operator  $K:J_{t_p}\times D_{p,m}\rightarrow H^{k-1}$ is
Lipschitz-continuous. Then  (\ref{E4.pr3}) implies
\begin{equation} H_\e(B_{p,m},H^{k-1}) \prec
\frac{1}{\e}\ln\frac{1}{\e}.
\end{equation}
Combining this with relation (\ref{E4.pr21}), we see that
$$ H_\e(\hat Q,H^{k-1})\succ \e^\nu H_\e(B_{p,m},H^{k-1}),$$
where $\nu>0$, which contradicts the inclusion $\hat Q \subset
B_{p,m} $.
\end{proof}

\bibliographystyle{plain}

\end{document}